\pdfoutput=1
\newif\ifFull
\Fulltrue
\ifFull
\documentclass[11pt]{article}
\usepackage[margin=1.25in]{geometry}
\else
\documentclass{llncs}
\renewenvironment{proof}{\noindent{\bf Proof:}}{\hspace*{\fill}\qed\bigskip}
\fi

\usepackage{graphicx,wrapfig}
\usepackage{amsmath,amssymb}
\usepackage{enumerate,verbatim}
\ifFull
\usepackage{mathptm}
\usepackage{amsthm}
\fi
\usepackage{hyperref}
\usepackage{cite}

\ifFull
\newtheorem{theorem}{Theorem}[section]
\newtheorem{lemma}[theorem]{Lemma}

\fi

\title{Confluent Hasse diagrams}
\ifFull
\author{David Eppstein \qquad Joseph A. Simons \\ 
 \\
 Department of Computer Science, University of California, Irvine, USA.}
\else
\author{David Eppstein \and Joseph A. Simons}
\institute{Department of Computer Science, University of California, 
	   Irvine, USA.}
\fi

\begin{document}

\maketitle

\begin{abstract}
We show that a transitively reduced digraph has a confluent upward drawing if
and only if its reachability relation has order dimension at most two. 
In this case, we construct a confluent upward drawing with $O(n^2)$ features, in an 
$O(n) \times O(n)$ grid in $O(n^2)$ time. For the digraphs representing series-parallel partial orders we show how to construct a drawing with $O(n)$ features in an $O(n) \times
O(n)$ grid in $O(n)$ time from a series-parallel decomposition of the partial
order. Our drawings are optimal in the number of confluent junctions they use.
\end{abstract}

\pagestyle{plain}

\section{Introduction}
One of the most important aspects of a graph drawing is that it should be readable: it should convey the structure
of the graph in a clear and concise way. Ease of understanding is difficult to quantify, so various
\ifFull
proxies for readability have been proposed; one of the most prominent is the number of edge
crossings.
That is, we should minimize the number of edge crossings in our drawing (a
planar drawing, if possible, is ideal), since crossings make drawings harder
to read. 
Another measure of readability
is the total amount of ink required by the drawing~\cite{Aeschlimann1992}. This measure can be
formulated in terms of Tufte's ``data-ink
ratio''~ \cite{Jourdan1995,Tufte1983}, according to which a large proportion of the ink on any
infographic should be devoted to information.
\fi
Thus given two different ways to
present information, we should choose the more succinct and crossing-free presentation.  

\begin{figure}[h]
\centering\includegraphics[width=.5\textwidth]{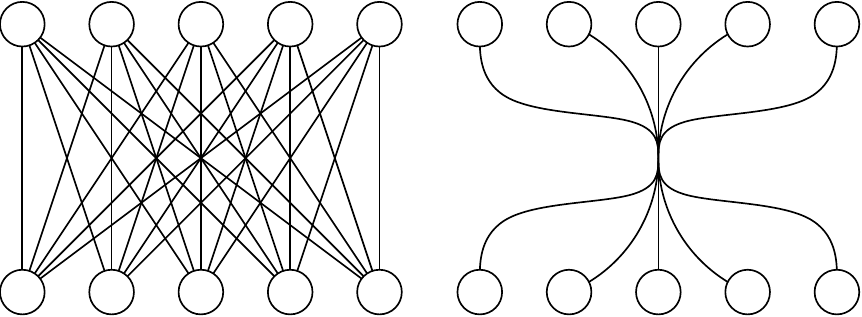}
\caption{\label{fig:k55}Conventional and confluent drawings of $K_{5,5}$.}
\end{figure}

Confluent drawing\cite{DicEppGoo-JGAA-05,EppGooMen-GD-05, EppGooMen-Algo-07,HirMeiRap-GD-07,HuiPelSch-Algo-07,QueAnc-GD-10,EppHolLof-GD-13} is a style of graph drawing in which multiple edges are combined into shared tracks, and two vertices are considered to be adjacent if a smooth path connects them in these tracks (Figure \ref{fig:k55}). This style was introduced to reduce crossings, and in many
cases it will also improve the ink requirement by representing dense
subgraphs concisely.  However, it
has had a limited impact to date, as there are only a few specialized graph
classes for which we can either guarantee the existence of a confluent drawing
or test for confluence efficiently.
A closely related graph drawing technique, edge bundling~\cite{Hol-TVCG-06,Gansner2011}, differs
from confluence in emphasizing the visualization of high level graph
structure, but does not necessarily seek to reduce the number of edge crossings.

\emph{Hasse diagrams} are a type of upward drawing of
transitively reduced directed acyclic graphs (DAGs) that have been used since the late
19th century to visualize partially ordered sets.
To maximize the readability of Hasse diagrams,
as with other types of graph drawing, we would like to draw them without crossings.  Thus upward planar graphs (DAGs that can be drawn so
that all edges go upwards and no edges cross) have been an important thread of
research in graph
drawing.
A DAG is upward planar if and
only if it is a subgraph of a planar st-graph, i.e. a planar DAG with one source and one sink, both on the
outer face~\cite{DiBattista1988}.
Testing upward planarity is NP-complete~\cite{Garg2002} but for DAGs with a single
source or a single sink it may be
tested efficiently~\cite{HutLub-SJC-96,Bertolazzi1998}.
However, many DAGs (even planar DAGs such as the one in Figure \ref{fig:simple}) are not upward planar.

\begin{figure}[t]
\centering\includegraphics[width=0.5\linewidth]{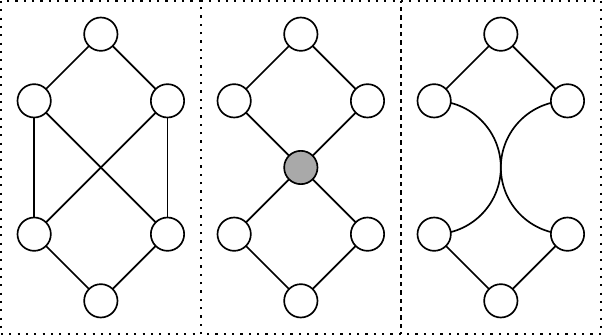}
\caption{
	\label{fig:simple}
	A simple DAG $P$ (left) that is not upward planar, although its underlying graph is planar. Its Dedekind--MacNeille completion (middle) is  upward planar, with an
	added element (shaded). Replacing that element
	with a junction creates an upward confluent drawing of $P$ (right).
	}
\end{figure}

In this paper, we bring these threads together by finding efficient algorithms for upward confluent drawing of transitively reduced DAGs.
We show that a graph has an upward confluent  drawing if and only if it represents a partial order $P$ with order dimension at most two, and that these drawings correspond to two-dimensional lattices containing $P$. We construct the smallest lattice containing~$P$ (its Dedekind--MacNeille completion) in worst-case-optimal time, and draw it confluently in area $O(n^2)$, using as few confluent junctions as possible. For series-parallel partial orders, the time and number of junctions can be reduced to linear.

\ifFull
Summarizing, we have the following new results:
\begin{itemize}
\item
We characterize the transitively reduced digraphs with
confluent upward drawings: they are the digraphs whose
reachability relation has order dimension at most two.
\item
We construct a confluent upward drawing for any
transitively reduced digraph that has one, by constructing the
Dedekind--MacNeille completion of the reachability poset and creating
confluent junctions corresponding to the added elements in the
completion. Our drawings have $O(n^2)$ junctions and track segments and can be embedded
into an $O(n) \times O(n)$ grid in $O(n^2)$ time. The number of junctions is the minimum possible for any confluent upward drawing of the given digraph.
\item
For series-parallel partial orders and the corresponding
transitively reduced graphs, we show how to construct a confluent drawing with
$O(n)$ elements, in an $O(n) \times O(n)$ grid, in $O(n)$ time given a
series-parallel decomposition of the partial order.
\end{itemize}
\fi

\section{Preliminaries}
\subsection{Posets and Lattices}
Here we review some basic definitions and notation concerning posets and lattices. 
For more, see e.g.~\cite{Birkhoff1967,t-cpos-92}.
A partially ordered set (partial order, or poset) $P = (V, \leq)$ is a set $V$ with a
reflexive, antisymmetric, and transitive binary relation $\leq$. We adopt the convention that $n = |V|$
unless otherwise stated. We also use $a < b$ to denote that $a
\leq b$ and $a \neq b$. 
We say that $a$ \emph{covers} $b$ in $P$ if $b < a$ and
$\nexists x \in P \mbox{ such that } b < x < a$.
Elements $a,b \in P$ are \emph{comparable} if $a \leq b$ or $b \leq a$;
otherwise, we write $a || b$ to indicate that they are \emph{incomparable}. A \emph{total order} or \emph{linear order} is a partial order in which every pair of elements in $P$ is comparable. If $R$ is a set of linear orders $R_i$, we can define a poset $P$ as the intersection of $R$: that is,  $a \leq b$ in $P$ if
and only if $a \leq b$ in every linear order $R_i$.  If $P$ can be defined from $R$ in this way, then $R$ is called
a \emph{realizer} of $P$. Every partial order $P$ has a realizer; the \emph{dimension} $\dim(P)$   is the 
smallest number of linear orders in a realizer of~$P$.

If $X \subseteq P$ is any subset of $P$, then an element $a \in P$ is called a
\emph{lower bound} of $X$ if it is less than or equal to every element of $X$.  Similarly, an element
$b$ is called an \emph{upper bound} of $X$ if it is greater than or equal to every element of $X$.
If $X$ has a lower bound~$a$ that belongs to $X$ itself, then $a$ is the (unique) \emph{least element} in $X$, and similarly if $X$ has an upper bound~$b$ that belongs to~$X$ then $b$ is the (unique) \emph{greatest element} in~$X$.
If the set $A$ of lower bounds of $X$ has a greatest element~$a$, then $a$ is the \emph{greatest lower bound} or \emph{infimum} of $X$, and similarly if the set $B$ of upper bounds of~$X$ has a lowest element $b$ then $b$ is the \emph{least upper bound} or \emph{supremum} of~$X$.
If $P$ itself has an infimum or a supremum, these elements are typically denoted by $0$~and~$1$ respectively.
If $P$ contains both an infimum and a supremum, it is said to be \emph{bounded}.

A poset $L$ is a \emph{lattice} if for every pair of elements $x$ and $y$ in $L$ the set $\{x,y\}$
has both an infimum and a supremum. In this context, the supremum of $\{x,y\}$ is
called the \emph{meet} of $x$ and $y$ and denoted $x \wedge y$, and similarly
the infimum is called the \emph{join} and denoted $x \vee y$. A lattice $L$ is
\emph{complete} if every subset of $L$ has an infimum and supremum in $L$. Every
finite lattice is complete and bounded.

\subsection{Hasse Diagrams and Upward Planarity}
Every poset $P = (V, \leq)$ can be represented by a directed acyclic
graph $G$ which has a vertex for each element in
$P$ and an edge $uv$ for each pair $(u,v)$ with $u\le v$ in $P$.
However, when we draw a poset it is more common to draw a different DAG, the \emph{transitive
reduction} $G'$ of $G$, in which there
is an edge from $u$ to $v$ in $G'$ if and only if $v$ covers $u$ in $P$. A \emph{Hasse
diagram} of $P$ 
is an upward drawing of $G'$, meaning that the $y$ coordinate of the head of each edge is
greater than the $y$ coordinate of the tail of each edge, and each edge is a
$y$-monotone curve, so that the drawing
``flows'' upward from smaller elements to larger elements. In a Hasse diagram,
we do not need to explicitly draw the edges as directed edges: the direction of
an edge is
represented implicitly by the relative position of its endpoints. There is an upward
path from $a$ to $b$ in a Hasse diagram of $P$ if and only if $a \leq b$.
A poset is \emph{planar} if it has a Hasse diagram that is upward planar, 
i.e. its transitive reduction has an upward drawing in which
none of the edges intersect except at a shared vertex. 

A finite lattice is planar if and only if its
transitive reduction is a planar st-graph, 
a DAG which contains exactly one source $s$ and one sink $t$ both of which belong to the outer face of an upward planar drawing~\cite{Platt1976}. More generally, any DAG is upward planar if and only if it is a subgraph of a planar
st-graph~\cite{DiBattista1988}. In the other direction, every planar finite bounded poset must be a lattice~\cite{Baker1972,Birkhoff1967,KellyRival1975}. This implies that a two-dimensional bounded poset that is not a lattice (such as the one on the left of Figure \ref{fig:simple}) cannot have an upward planar drawing, and that planarity (a crossing-free drawing) and two-dimensionality (realization by a pair of linear orders) are distinct for non-lattice posets.

\subsection{Lattice Completion of a Poset}

The Dedekind--MacNeille completion of a poset
\ifFull
$P$ (also called the normal completion or the completion)
\fi
is the smallest complete lattice containing~$P$~\cite{Mac-TAMS-37}.
\ifFull
Its construction is based on  Dedekind's construction of
the real numbers as Dede\-kind cuts of rational numbers.
\fi
For any subset $X$ of $P$, let $X^-$ and $X^+$ denote the set of lower bounds and
upper bounds of $X$ respectively. A \emph{cut} of $P$ is a pair $A,B \subseteq
P$ such that $A^+ = B$ and $A = B^-$;  the completion of $P$ has these cuts as its elements. The completion is partially ordered by set containment: if $(A,B)$ and $(C,D)$ are cuts, then $(A,B)\le (C,D)$ if and only if $A\subseteq C$ and $B\supseteq D$.
The element of the completion corresponding to an element $x$ of $P$ is the cut $(\{x\}^-,\{x\}^+)$, and the new elements added to $P$ to make it into a lattice come from cuts $(A,B)$ for which $A\cap B=\emptyset$.
The completion automatically has the same dimension as the partial order from which it was constructed~\cite{Nov-MA-69}.

Ganter and Kuznetsov~\cite{Ganter1998}
give a stepwise algorithm for constructing the completion of $P$.
Given a poset $P$ and its completion $L$ they show how to complete a one-element
extension of $P$ in time $O(|L|\cdot |P| \cdot \omega(P))$, where $\omega(P)$ denotes
the width of $P$. To compute the completion of a large poset, they begin with a single-element poset (whose completion is trivial) and use this subroutine to add elements one at a time;
therefore, the total time is $O(|L|\cdot |P|^2 \cdot \omega(P))$.
Nourine and Raynaud~\cite{Nourine1999}
give an algorithm with running time $O((|P| + |B|)\cdot |B| \cdot |L|)$ where
$B$ is a \emph{basis} of $P$ (a set of subsets of $P$ which generate $L$).
As part of our drawing algorithm, we improve these results in the case of two-dimensional posets: we show for such sets how to construct the completion in time $O(|P|^2)$, optimal in the worst case since (as we also show) there exist two-dimensional posets whose completion has a quadratic number of elements.

\subsection{Confluent Drawing}
Confluent drawing is a technique for drawing non-planar diagrams without crossings~\cite{DicEppGoo-JGAA-05,EppGooMen-GD-05,EppGooMen-Algo-07,HirMeiRap-GD-07,HuiPelSch-Algo-07,QueAnc-GD-10,EppHolLof-GD-13} by merging together groups of edges and drawing them
as \emph{tracks} that, like train tracks, meet smoothly at junction points but  do not cross. 
A \emph{confluent drawing} consists of a set of labeled points (\emph{vertices} and \emph{junctions}) and curves (\emph{track segments}) in the Euclidean plane, such that the two endpoints of each track segment are vertices or junctions, such that no two track segments intersect except at a shared endpoint, and such that all track segments that meet at a junction share a common tangent line at that point. The graph represented by a confluent drawing has as its vertices the vertices of the drawing; two vertices $u$ and $v$ are adjacent if and only if there is a smooth curve in the plane from $u$ to $v$ that is a union of track segments and that does not pass through any other vertex. (Some papers on confluence require that this curve also be non-self-intersecting but that requirement is irrelevant for upward drawings since monotone curves cannot self-intersect.)
An undirected graph $G$ is \emph{confluent} if and only if there exists a confluent drawing that represents it.

We define a \emph{confluent diagram} of a poset to be a drawing of its transitive reduction in a way that is both confluent and upwards. In other words, if $G$ is a directed acyclic graph representing a poset $P$, 
then we define a confluent diagram of $P$ to be an upward confluent drawing of the
transitive reduction of $G$ in which all tracks are oriented upwards (monotonic in the $y$ direction), and therefore all smooth curves passing through the tracks are similarly oriented. 
For each pair of elements $a,b \in P$, the drawing should have a smooth track from $a$ upwards to $b$ if
and only if $a$ is covered by $b$. We also require that for each source there exists an unbounded $y$-monotone curve downwards that does not cross the diagram -- that is, that each source can be seen from below -- and symmetrically that each sink can be seen from above. In the application to visualization of partial orders, this is a natural restriction as it makes the minimal and maximal elements easy to find in the drawing.

%

\section{Drawing Posets of Dimension Two}\label{sec:alg}
Let $G$ be a poset with dimension at most two. 
We now describe an $O(n^2)$ algorithm to embed a confluent diagram of $P$ in an
$O(n) \times O(n)$ grid.
That is, we will generate an upward confluent drawing of the transitive
reduction of a DAG representing $P$ such that each vertex in the drawing has
integer coordinates.  

\begin{figure}[t]
\centering\includegraphics[width=0.9\linewidth]{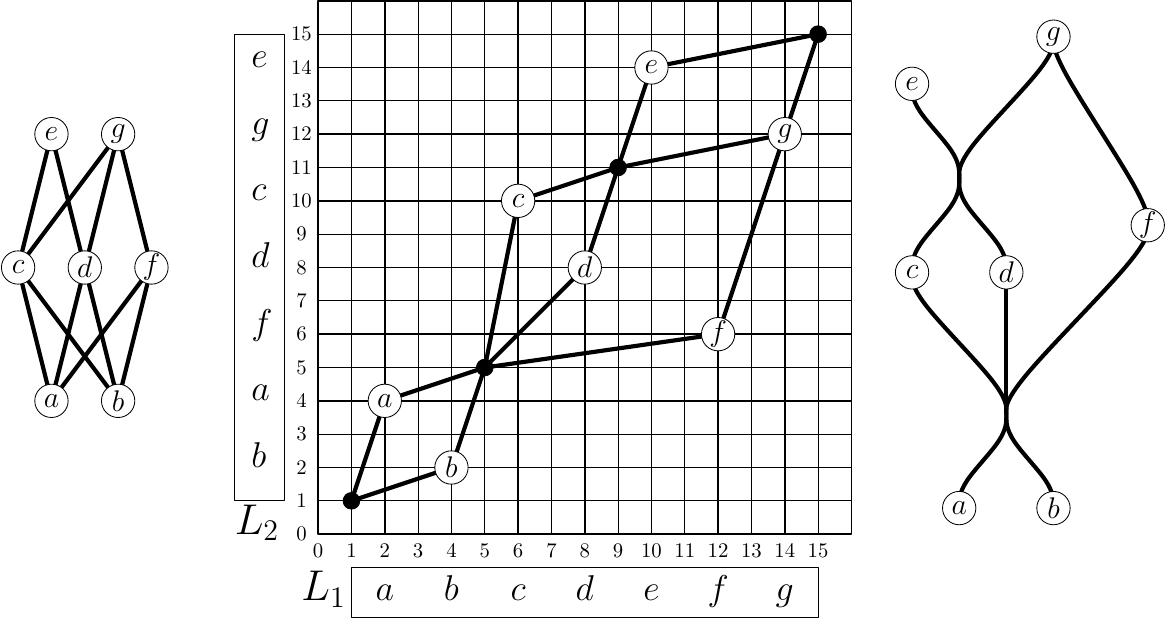}
\caption{Example of our algorithm. \emph{Left}: Input poset~$P$.
\emph{Middle}: Grid embedding with added points and dominance pairs.
\emph{Right}: Completion points replaced by confluent junctions and rotated $45^\circ$.}
\label{fig:algEx}
\end{figure}

Our algorithm has three phases. In the first phase, we embed the elements of $P$
in a $(2n+1) \times (2n+1)$ grid. Recall that since $P$ has dimension two, it is
realized by two linear orders, which correspond to two different total orderings of
the same $n$ elements in $P$.  Thus, the first steps of our algorithm are:
\begin{enumerate}
\item[1.]
    \begin{enumerate}
	\item
		Find two linear orders $L_1$ and $L_2$ that realize
		$P$. This can be done
		 in $O(n^2)$ time from any graph whose transitive closure is $P$  by Algorithm 1 of
         Ma and Spinrad~\cite{MaSpinrad}.
	\item
		For each element $p$ of $P$, having position $i$ in $L_1$ and $j$ in
        $L_2$ with $1\le i,j \le n$, place a vertex representing $p$ in
		the grid with coordinates $(2i, 2j)$.	
    \end{enumerate}
\end{enumerate}
After this step, the even rows and columns in the grid each contain
exactly one element of $P$, and the dominance relationship of these points
corresponds to the order of the elements in $P$. Recall that for two
elements $p$ and $q$ in the plane, $p$ \emph{dominates}
$q$ if and only if $p_i \geq q_i$ for each coordinate $i$ and $p \neq q$.

In the second phase, we insert additional points representing elements of
the completion of $P$;  these completion nodes correspond to confluent junctions
in the confluent diagram of $P$. We defer to Section~\ref{sec:correctness} the proof that the dominance order on the points generated in the first two phases gives the completion of~$P$.
\begin{enumerate}
\item[2.]
	For each pair of odd indices $(i,j) \in [3, 2n - 1]^2,$ 
	insert a junction in
	the grid with coordinates $(i, j)$ if all of the following four
	conditions hold: 
	\begin{itemize} 
	\item 
		The poset point with $x$-coordinate $i-1$ has $y$-coordinate less than $j-1$.  
	\item 
		The  point with $x$-coordinate $i+1$ has $y$-coordinate greater than $j+1$.  
	\item 
		The  point with $y$-coordinate $j-1$ has $x$-coordinate less than $i-1$.  
	\item
		The  point with $y$-coordinate $j+1$ has $x$-coordinate greater than $i+1$.
	\end{itemize} 
	In addition if $P$ does not already have a least or a greatest
element, then insert invisible points at $(1,1)$ and $(2n+1,2n+1)$
respectively.
\end{enumerate}

In the third phase, we generate the segments of the confluent
diagram. These segments correspond to direct dominance pairs
of  points from the first two phases. It is possible to find all dominance pairs
in a set of $N$ points in time $O(N \log N + k)$~\cite{GutHarNur-Algs-89} where
$k$ is the number of dominance pairs, but in our case
\ifFull
$N$ may be too large, so
\fi
this would only lead to an $O(n^2 \log n)$ time bound. Instead, we leverage the fact that the vertices are embedded in an $O(n) \times
O(n)$ grid, and use the following $O(n^2+k)$ time method to generate dominance
pairs using a stack-based algorithm related to Graham scan within each row. We
prove later that the diagram is planar and therefore that the number $k$ of
dominance pairs is $O(n^2)$.
\begin{enumerate}
\item[3.]
    For each column $c$ we maintain a value $t_c$, the topmost element seen so far in column~$c$.
    Initialize each $t_c$ to \texttt{None}.
      
    Then, for each row $r$ from $1$ to $2n+1$:
    \begin{enumerate}
        \item Initialize an empty stack $S$.
        \item For each column $c$ from $1$ to $2n+1$:
            \begin{enumerate}
                \item If there is a vertex or junction $p$ at $(r,c)$, add an edge from every element of $S$ to $p$, add an edge from  $t_c$ to $p$ (if $t_c$ is is not \texttt{None}), and set $t_c$~to~$p$.
                \item If $t_c$ is is not \texttt{None}, pop all items from $S$ whose row number is less than or equal to the row number of $t_c$, and push $t_c$ onto $S$.
            \end{enumerate}
    \end{enumerate}
\end{enumerate}

Thus we have computed the coordinates of all elements, confluent junctions, and
edges in the confluent diagram. When we
render the drawing, we rotate it $45^\circ$ counterclockwise to make it upward confluent (Figure \ref{fig:algEx}).

\begin{figure}[t]
\centering
\includegraphics[width=2.25in]{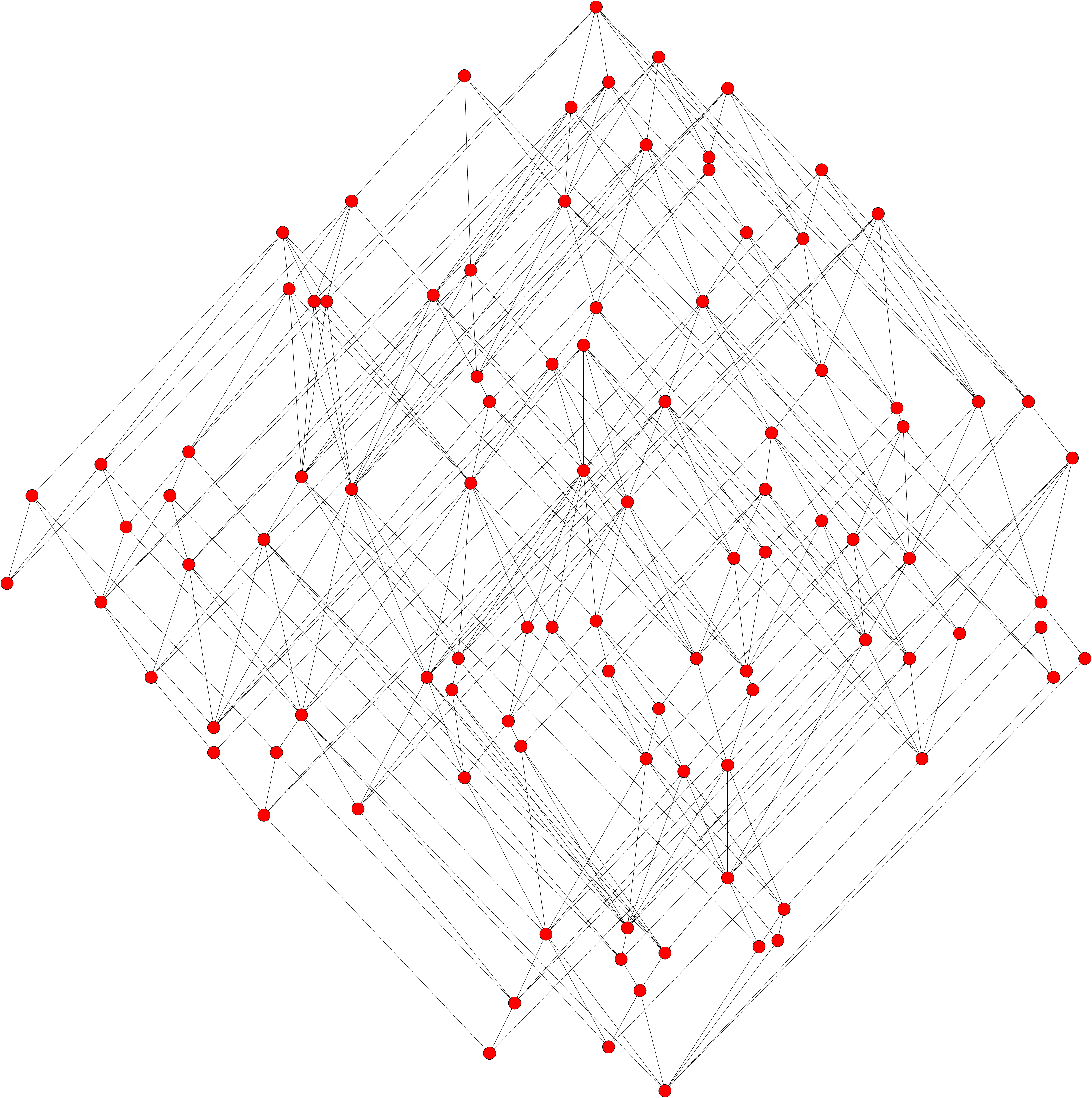}\quad
\includegraphics[width=2.25in]{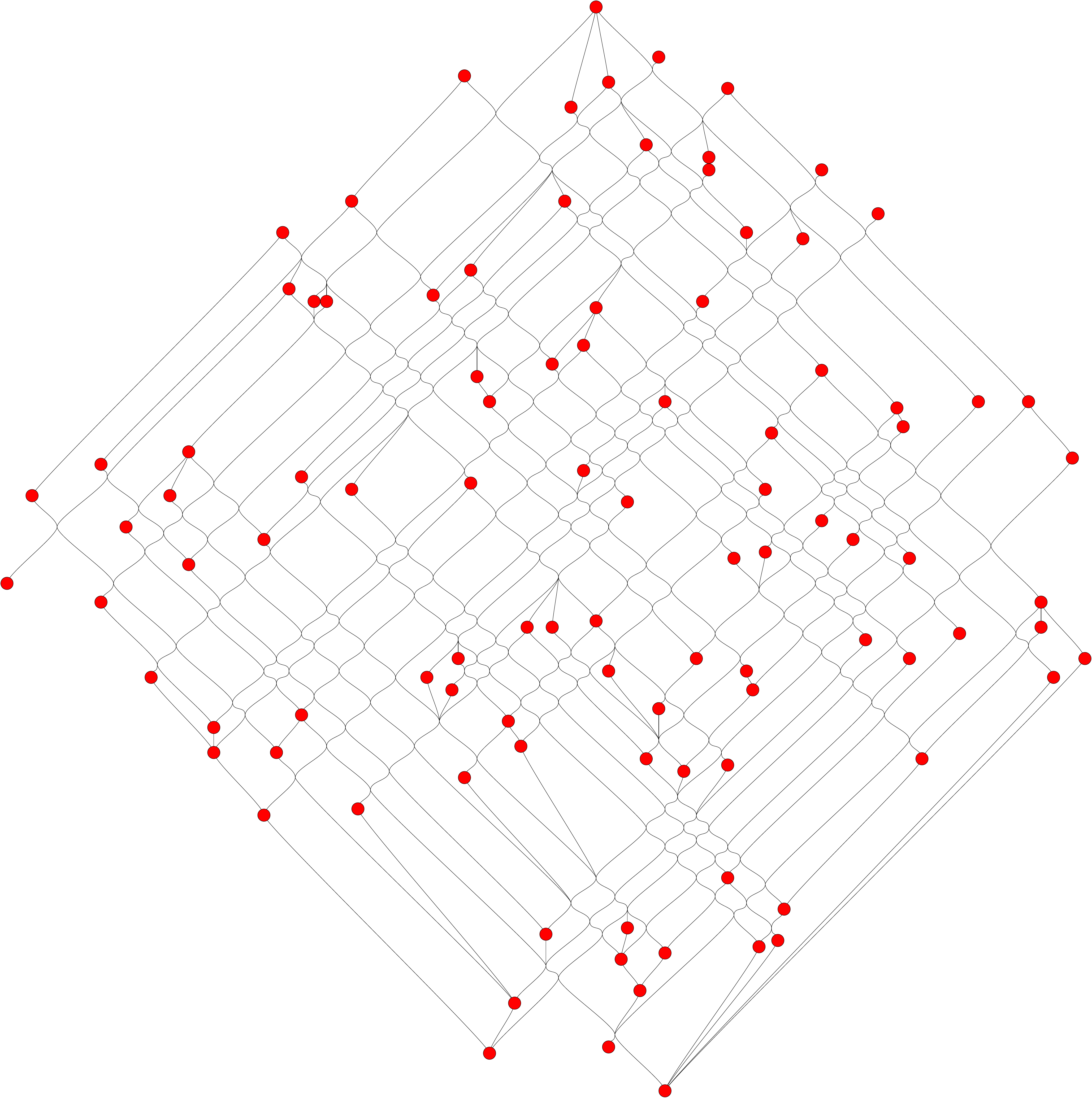}
\caption{A 100-element partially ordered set, the intersection of two random permutations, drawn  as a conventional Hasse diagram with crossings (left), and as a confluent Hasse diagram (right).}
\label{fig:implementation} 
\end{figure}

Examples of non-confluent and confluent drawings of the same 100-element set are shown in Figure~\ref{fig:implementation}. Our Python implementation renders the confluent track segments as cubic B\'ezier curves with control points at a small fixed distance directly above and below each confluent junction. Two such curves cannot cross each other: for pairs of edges that do not share an endpoint, this follows from the fact that the convex hulls of the control points are disjoint and that the curves lie within the convex hulls, while for pairs of curves that share an endpoint it follows from the fact that the two curves are images of each other under an affine transformation of the plane and that (for pairs of edges sharing an endpoint) the direction that any point on the curve is translated by this affine transformation is transverse to the tangent direction of the curve at that point.

If the input is provided as a realizer rather than as a graph, and its completion has few elements, then it is  possible to construct the diagram more efficiently. To do so, construct for each odd-indexed row or column of the integer grid an axis-parallel line segment that passes through a grid point if and only if that point meets two of the four conditions for adding a junction in phase~two of our algorithm. The junctions can be recovered as the intersections of these line segments, and we may compute the edges of the diagram using an output-sensitive algorithm for dominance pairs. By using integer searching data structures the total time for this algorithm may be reduced to $O((n+k)\log\log n)$, where $k$ is the number of confluent junctions; we omit the details.

\section{Algorithm Correctness and Minimality}
\label{sec:correctness}
In this section we prove that the algorithm of Section
\ref{sec:alg} is correct and has optimal running time. Our analysis also shows that a poset $P$ has a confluent diagram if and only if it has
dimension at most two.  

\begin{lemma}[Baker, Fishburn and Roberts~\cite{Baker1972}]
\label{lem:planarDim2}
Let $P$ be a bounded finite planar poset. Then $P$ is a lattice and has dimension at most 2.
\end{lemma}

\begin{lemma}
\label{lem:junctionToLattice}
Let $P$ be a finite poset with a confluent Hasse diagram $D$.
Then $\dim(P) \leq 2$, and there exists a two-dimensional lattice $C$ 
containing $P$ such that the elements of $C\setminus P$ (other than the top and bottom element, if they do not belong to $P$) correspond one-for-one with the 
confluent junctions of $D$.
\end{lemma}

\begin{proof}
Replace the confluent junctions of $D$ with vertices, and re-interpret the confluent segments as edges between these vertices. If there is more than one minimal vertex of $P$, add a vertex below all minimal vertices, connected to the minimal vertices by upward edges, and similarly if there is more than one maximal vertex of $P$, add a vertex above all maximal vertices connected to them by edges. The modified drawing is st-planar and hence by Lemma~\ref{lem:planarDim2} represents a lattice, which clearly contains $P$.
\end{proof}

\begin{lemma}
\label{lem:completion}
Let $P$ be a finite poset with order dimension at most two, let $C$ be the completion of $P$, and let $S$ be the set of elements of $C\setminus P$ (other than the top and bottom element, if $P$ itself is not bounded). Then the elements of $S$ coincide with the junction points added in phase 2 of our algorithm, and the dominance ordering on these points coincides with the lattice ordering in $C$.
\end{lemma}

\begin{proof}
In one direction, let $p$ be a junction point added in phase 2 of our algorithm, and $p^-$ and $p^+$ be the sets of points from phase 1 that are dominated by $p$ and that dominate $p$ respectively.
Then it follows from the four conditions according to which phase 2 adds a point that $(p^-,p^+)$ forms a cut in $P$. The equivalence of the dominance and lattice orderings on pairs consisting of a junction point and a point from $P$ follows immediately, and the same equivalence for pairs of junction points is also easy to verify.

In the other direction, we must show that we add a junction point for every element of $S$, that is, every cut $(L,U)$ where $L$ has more than one maximal element and $U$ has more than one minimal element. Let $i$ be one less than the minimum $x$-coordinate of a point in $U$, and let $j$ be one less than the minimum $y$-coordinate; then (because the coordinates of points in $P$ are their positions in the two orderings of a realizer) the set $L$ of points dominated by every point in $U$ equals the set of points below and to the left of $(i,j)$. Two of the four conditions of phase 2 are automatically met at~$(i,j)$: the points with $x$-coordinate $i+1$ and with $y$-coordinate $j+1$ are both in $U$ and are distinct because $U$ has more than one minimal point. The other two conditions must also be met, for if they were not then the point violating the condition would dominate $L$, contradicting the fact that all points that dominate $L$ belong to $U$.
\end{proof}

\begin{theorem}
A given partial order $P$ has a confluent diagram if and only if $\dim(P) \leq
2$. If $P$ has a confluent diagram,
the algorithm of Section \ref{sec:alg} computes a valid confluent diagram of
$P$, and embeds that diagram in a $O(n) \times O(n)$ grid in worst case optimal $O(n^2)$ time. The number of confluent junctions in the drawing is the minimum possible for any confluent diagram of $P$.
\end{theorem}

\begin{proof}
If a poset $P$ has dimension three or more, then so does any lattice containing it, and by Lemma~\ref{lem:planarDim2} and Lemma~\ref{lem:junctionToLattice} there can be no confluent diagram of $P$. Otherwise, we may assume that $P$ has dimension at most two.
\begin{figure}[t]
\centering
\includegraphics[width=2.25in]{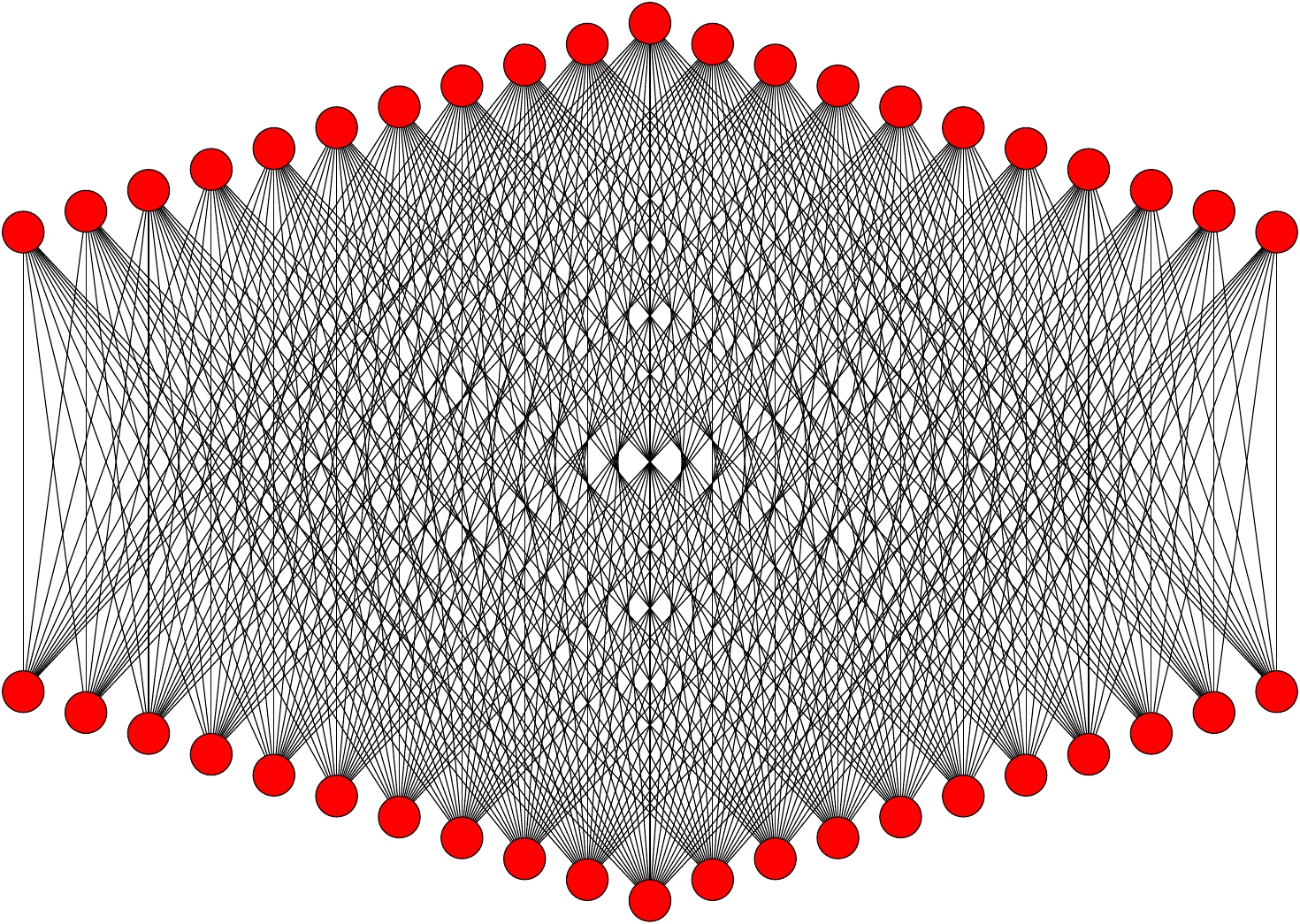}\quad
\includegraphics[width=2.25in]{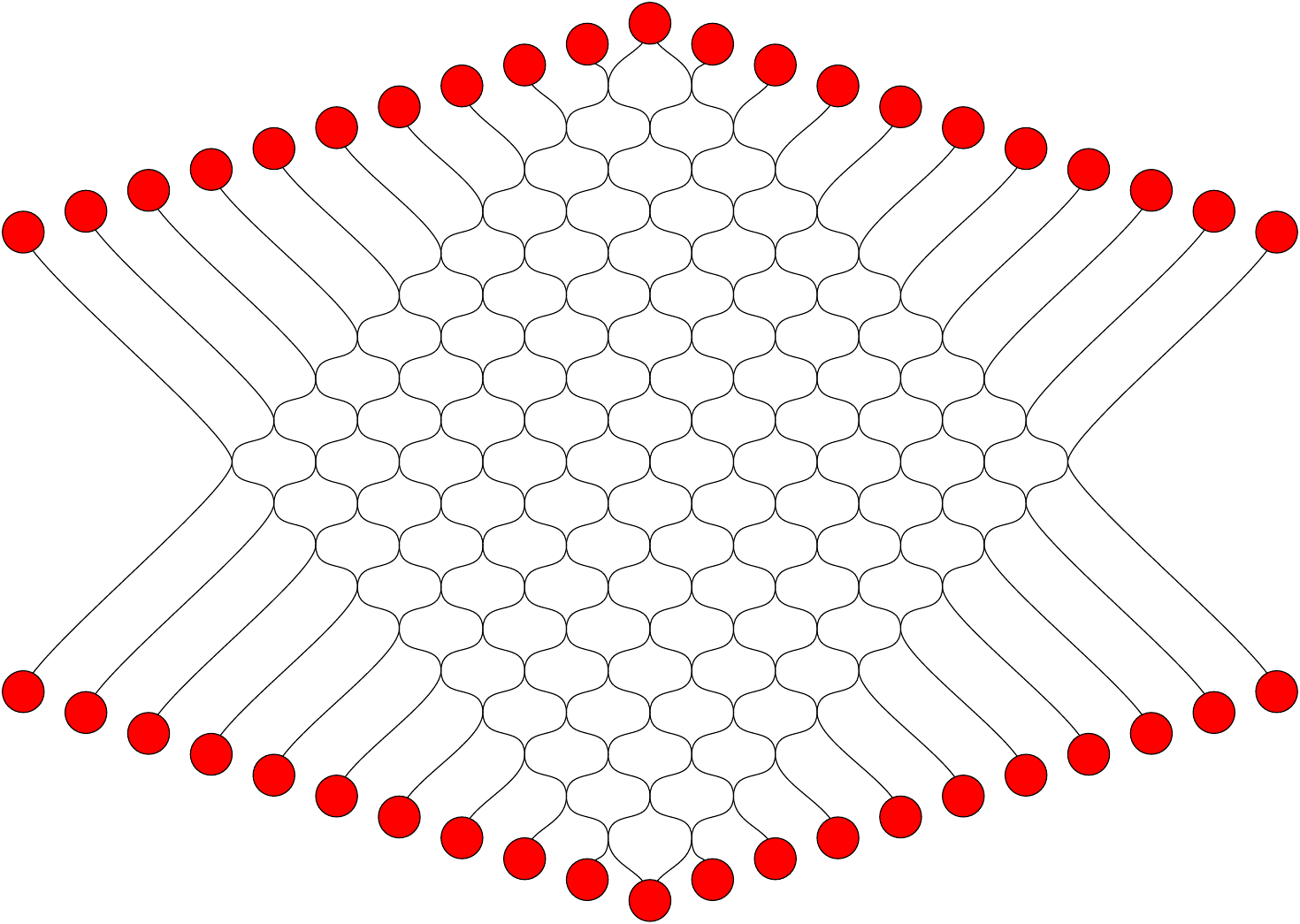} 
\caption{A poset $P$ with $O(n)$ elements and dimension 2 whose completion
has size $\Omega(n^2)$. On the left is the normal Hasse diagram, and on the
right is the confluent version as drawn by our algorithm. The two permutations $L_1$
and $L_2$ generating $P$ are the identity and the permutation
$(3n,3n-2,\dots,n;4n+1,n-1,4n,n-2,\dots,3n+2,0;3n+1,3n-1,\dots,n+1)$.}
\label{fig:lowerBoundPoset}
\end{figure}

By Lemma~\ref{lem:completion}, the dominance ordering on the points computed by our algorithm coincides (except possibly for the removal of the top and bottom elements) with the completion of $P$. In this set of points, there can be no crossing pairs of dominance relations, for if the edges $(L_1,U_1)$--$(L_2,U_2)$ and $(L_3,U_3)$--$(L_4,U_4)$ crossed (where $(L_i,U_i)$ is a cut either added in the completion or corresponding to an original point of $P$) then $(L_1\cup L_3,U_2\cup U_4)$ would also be a cut whose point would lie between the other four points, contradicting the assumption that these edges represent minimal dominance pairs.
Therefore, the diagram constructed by our algorithm is planar, and by Lemma~\ref{lem:planarDim2} it must represent a lattice superset of $P$. The  added elements belong to the completion, so the diagram must represent a subset of the completion, and since the completion has no proper lattice subsets it must represent the completion itself. The completion gives the minimum number of added elements (and therefore, by Lemma~\ref{lem:junctionToLattice}, the minimum number of junctions) of any diagram for $P$.

Our algorithm spends $O(n^2)$ time in its first two phases as it iterates over $O(n^2)$ grid cells spending constant time per cell. In the third phase, it uses constant time per edge and by planarity there are $O(n^2)$ edges, so the time is again $O(n^2)$. This time bound is optimal since (as shown in Figure~\ref{fig:lowerBoundPoset}) there exist two-dimensional posets whose completion has $\Omega(n^2)$ elements.
\end{proof}

Although our method produces drawings in a grid of linear dimensions, it may be possible in some cases to compact our drawings into a smaller grid. An algorithm of de la Higuera and Nourine~\cite{HigNou-TCS-97} may be used to find the smallest grid into which a drawing produced by our algorithm can be compacted.

Subsequent to our work, a different embedding into lattices has been applied by Cz\'edli~\cite{Cze-12} to characterize the partial orders of dimension two as being the posets with \emph{quasiplanar Hasse diagrams}, diagrams in which each incomparable elements has one element on a consistent side of all maximal chains through the other element. The lattices into which Cz\'edli embeds a partial order are semimodular, a property that does not hold for all two-dimensional lattices. Therefore, unlike the Dedekind--MacNeille completion that we use, these lattices do not have a minimal number of elements: a non-semimodular two-dimensional lattice will be drawn with no additional confluent junctions by our algorithm, but will be augmented by additional elements in the method of Cz\'edli.

\section{Confluent Drawings of Series-Parallel Posets}

\begin{figure}[t]
\centering\includegraphics[width=.5\textwidth]{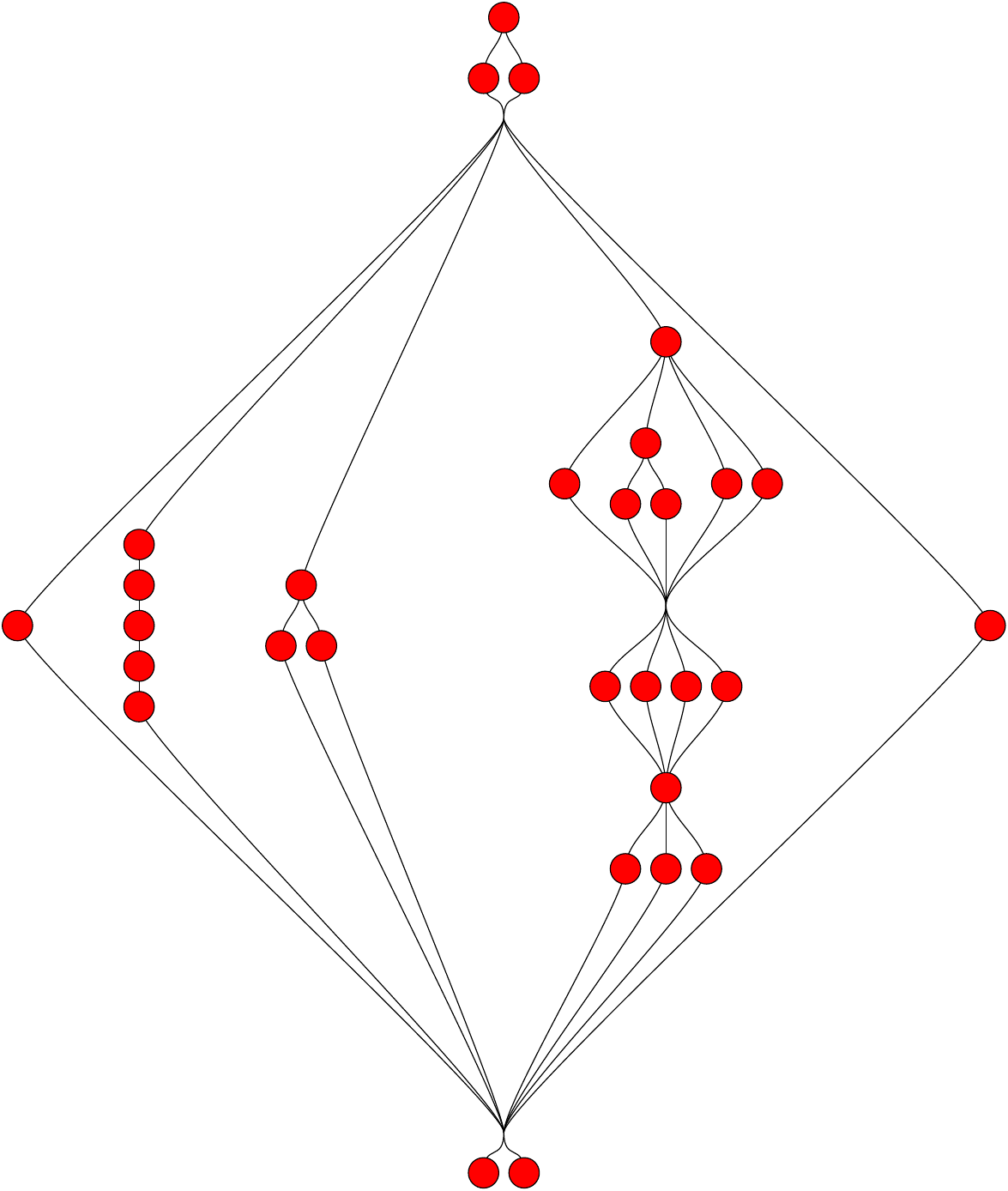}
\caption{A series-parallel poset.}
\label{fig:sp}
\end{figure}

A \emph{series-parallel partial order} is a poset that can be built up from single elements by two simple composition operations:
\begin{itemize}
\item The \emph{series composition} $P;Q$ of  posets $P$ and $Q$ is the order on the set $P\cup Q$ in which $p\le q$ for every $p\in P$ and $q\in Q$.
\item The \emph{parallel composition} $P||Q$ is the order on $P\cup Q$ in which
$p || q$ for every $p \in P$ and  $q \in Q$.
\end{itemize}

Pairs of elements that are both from $P$ or both from $Q$ retain their ordering in the larger set.


 Series-parallel partial orders are attractive because many important
 computational problems can be solved more easily in them than in more general
 posets, and because they have applications to a wide variety of problems
 including scheduling \cite{Mohring1999}, concurrency \cite{Lodaya1998}, 
 data mining \cite{Mannila2000},  networking
 \cite{Amer1994}, and more (see \cite{Mohring}).

Series-parallel partial orders can be represented naturally by a binary
tree, known as a decomposition tree of the order. The leaves of the tree correspond
to single element sets and the internal nodes of the tree correspond to series
or parallel composition operations.  
As the following theorem shows, given a decomposition tree $T$ for a
series-parallel partial order $P$, we can draw the confluent diagram of $P$ in
linear time by traversing $T$, performing the corresponding composition
operations, and inserting confluent junctions when necessary. 
\begin{figure}[t]
\centering
\includegraphics[width=0.8\linewidth]{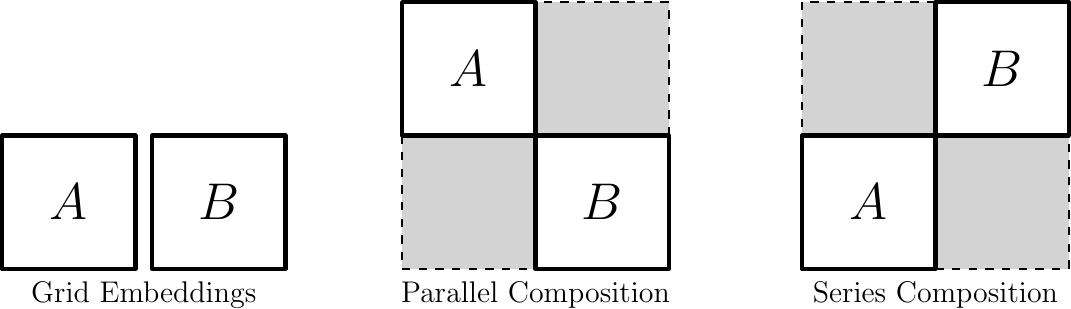}
\caption{Series and parallel composition operations on two drawings $A$ and $B$.
\label{fig:comps}}
\end{figure} 
\begin{theorem}
Let $P$ be a series-parallel partial order, given as its decomposition tree. Then a confluent diagram of $P$ with a
linear number of junctions  can be drawn in an $O(n) \times O(n)$ grid in linear time.
\end{theorem}

\begin{figure}[t]
\centering
\includegraphics[width=0.8\linewidth]{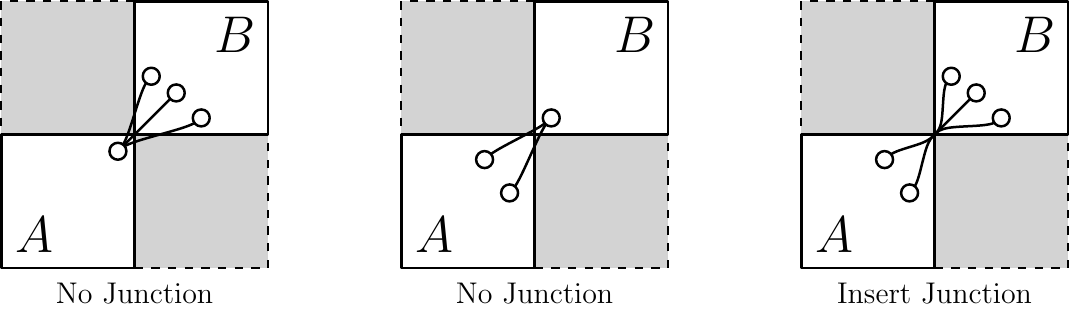}
\caption{Series
composition $A;B$ has a confluent junction
if and only if $A$ has no unique upper bound and $B$ has no unique lower
	bound. 
\label{fig:addConfluent}}
\end{figure}

\begin{proof}
We traverse the decomposition tree in post-order, recursively
finding embeddings for each subtree.
For each tree node, we do the following:
\begin{enumerate}
\item 
	If the node is a leaf, then we embed the corresponding element in a single
	grid cell
\item
	Otherwise, if the node is a series or parallel node, then we translate the grid embeddings of its two children so that their bounding boxes meet corner to corner (Figure \ref{fig:comps}).
\item
	For a series composition $A;B$ we also
	insert a confluent junction at the shared corner of $A$ and $B$  if and only if $A$ has more than one
	maximal element and $B$ has more than one minimal element (Figure
	\ref{fig:addConfluent}). 
\end{enumerate}
By using a linked list of the maximal and minimal nodes for the current
subtrees, we can perform these operations in time proportional to the number of leaves in
the decomposition tree. Therefore the total time is linear.
The size of the grid will be proportional to the size of the decomposition tree, i.e., $O(n) \times O(n)$
\end{proof}

\section{Experiments}

As a proof of concept for our method, we implemented it and tested how well it performs, in terms of the number of edges or confluent segments drawn and the ink usage of our drawings.

In our experiments,
we consider drawing two classes of partial orders separately, first, 
the special case of series-parallel partial orders and second, 
all two-dimensional partial orders. 
We consider several different sizes, and for each size and class we generate
graphs of that class and size uniformly at random. 
We calculate
the number of edges and total edge length (ink) in the traditional Hasse
diagram and confluent Hasse diagram corresponding to each graph. In the
traditional Hasse diagram, each edge is drawn as a straight line between two
vertices. In the
confluent diagram, an ``edge'' between two vertices may go through multiple
confluent junctions, and multiple ``edges'' may reuse the same curve incident to
a junction. Thus, we count the number of edges in the confluent diagram as the number of
confluent segments; we define a
\emph{confluent segment} as a curve
between \emph{endpoints}, where each endpoint is either a vertex or a confluent
junction. Each segment is
drawn as a cubic bezier curve, but for practical reasons we approximate its length 
as the length of the three line segments through its control points
(Figure~\ref{fig:inkMeasure}). Note that
this measure will never underestimate the ink used by any edge in our confluent diagram.
Because of the quadratic growth in the output complexity of some of our drawings, we limited our experiments  to graphs of at most 2048 vertices.
For each of the smaller graph sizes, 10,000 permutations were generated uniformly at random. 
For the larger sizes, we only generated 1000, because of the long run-times,
and the low variance of the results.

\begin{figure}[t]
\centering
\includegraphics[width=0.2\linewidth]{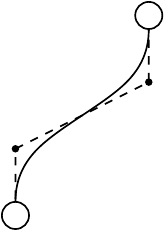}
\caption{Ink used by each cubic bezier segment is approximated by the straight-line
    path through its control points.
\label{fig:inkMeasure}}
\end{figure}

In Figure~\ref{fig:ratio-ser-num-ln} and Figure~\ref{fig:ratio-ser-ink-ln}, 
we compare the average edge count and average ink used for 
traditional and confluent Hasse diagrams of series-parallel partial orders. 
The result is that confluent drawings are consistently better 
than the traditional Hasse diagram in both number of edges and ink used for this class of inputs.

In Figures~\ref{fig:ratio-rand-num} and~\ref{fig:ratio-rand-ink-ln}, 
we compare the average edge count and average ink used for 
traditional and confluent Hasse diagrams 
of two-dimensional partial orders.
The result fpr these inputs is that on average the confluent Hasse diagram uses substantially less ink than the
traditional Hasse despite the fact that it contains many more edges. Thus, although the
confluent Hasse diagram is more complex to render it is also dramatically easier to read.  

The reason random two-dimensional partial orders use more edges in a confluent drawing than in a traditional Hasse diagram is that they have a large number of confluent junctions.
As we show below, for two-dimensional partial orders generated
uniformly at random, the expected number of edges in a Hasse diagram is $\Theta(n
\log n)$, but the expected number of confluent junctions in a confluent diagram is
$\Theta(n^2)$. The large number of confluent junctions necessarily implies 
a drawing with a large number of edges. However, each confluent junction reduces
the visual clutter of at least one edge crossing. Thus, while a large number of junctions
indicates a drawing with a large number of edges, it also indicates a drawing 
that is substantially easier to read than the corresponding traditional Hasse
diagram. 

\paragraph{Graph Generation}
Input graphs were sampled uniformly at random from the set of all input graphs
for each given size and type (two-dimensional or series-parallel partial
orders). 

We generate each two-dimensional partial order of size 
$n$ by generating a permutation on $n$ elements uniformly at random. 
Each permutation $\pi$ maps the set of elements $L_n = [1,n]$ in sorted order to some
other order $\pi(L_n)$. Thus, we have a pair of linear orders $L_n,
\pi(L_n)$ which define a
partial order, since each two-dimensional partial order 
is realized by pair of linear orders, and by relabeling the elements,
any pair of size $n$ linear orders corresponds to $L_n$ and some
permutation of $L_n$.

The set of series-parallel partial orders is a subset of the two-dimensional
partial orders. Thus, we generate each series-parallel partial order by sampling
uniformly at random from only those permutations that correspond to a
series-parallel partial order. Such a permutation can be decomposed uniquely into series and parallel compositions with the constraint that the left argument of each series decomposition is parallel (or an atom) and the left argument of each parallel decomposition is series (or an atom). By means of this decomposition, we may count the number of permutations whose outer composition operation is series (so that the Hasse diagram is connected) by the recurrence relation
$$C_n=C_{n-1}+\sum_{i=1}^{n-2} 2C_i C_{n-i}.$$
The numbers generated by this recurrence are called the little Schroeder numbers. Here $i$ is the size of the left argument of the outer series composition, $n-i$ is the size of its right argument, and the factor of $2$ accounts for the choice of whether to use a series or parallel composition in the right argument. When the right argument has only one element (i.e., when $n-i=1$), this choice is irrelevant, so the term for $i=n-1$ omits the factor of two and is pulled out of the sum as $C_{n-1}$.
Our algorithm for generating a random order chooses a random integer in the range from $0$ to $C_n-1$, compares it to the partial sums of the terms on the right hand side of the recurrence to determine which value of $i$ to use, and returns the concatenation of two randomly generated permutations of sizes $i$ and $n-i$ (with the first of these two permutations always reversed so that its outer composition operation is parallel rather than series, and the second reversed with probability $1/2$).

\paragraph{Expected edge count.}
Eckhardt~\emph{et~al.}~\cite{emn-flcacd-07}
 show that the expected number of edges in a transitively reduced digraph is
$\Theta(n \log n)$ in a random graph model where each edge is included in the graph with
probability $p$. Our model for generating the graphs is somewhat different, but
leads to the same asymptotic bound.
There is a clear bijection between any permutation on $n$ elements and a
two-dimensional partial order of $n$ elements. Thus, we generate a
permutation uniformly at random, which corresponds exactly to a two-dimensional partial
order. Under this model each element in the partial order 
has a pair of coordinates equal to the index
and value of the corresponding element in the permutation.   
There is an edge $(u,v)$ in the Hasse diagram if and only if vertex $v$ covers vertex $u$.
That is, $v$ dominates $u$ 
($\textit{u.index} < \textit{v.index}$ and $\textit{u.value} < \textit{v.value}$), and there
is no third vertex $z$ such that $z$ dominates $u$ and $v$ dominates $z$.

Let $e_k$ be the element with index $k$ in column $k$. For each element $e_{j+k}$ in indices
$j+k$, $j \in [1,n-k]$, $e_j$ covers $e_k$ if it is the successor of $e_k$ among the $j+1$
elements in indices $[k, k+j]$ ordered by value. Thus, the probability that $e_{j+k}$ covers
$e_k$ is $1/(j + 1)$. 

Let $C_k$ denote the number of elements which cover element $e_{n-k}$.
\[
E[C_k] = \sum_{j=1}^k \frac{1}{j + 1} = \sum_{i=2}^{k+1} \frac{1}{i} =
         \sum_{i=1}^{k+1} \frac{1}{i} -1 = H_{k+1 - 1}
\]

Thus, by linearity of expectation, the expected number of edges in a traditional drawing of a random two-dimensional partial order is
\[
\sum_{k=0}^{n-1} E[C_{k}] = \sum_{k=0}^{n-1} H_{k+1} - 1 = \sum_{k=1}^{n} H_k - 1 =
\Theta(n \log n) 
\]

\paragraph{Expected number of confluent junctions.}
Note that each even row and each even column in $[2, 2n]$ 
in the grid contains exactly one point.
Let $x_j$ denote the $x$-coordinate of the poset point with $y$-coordinate $j$ and 
let $y_i$ denote the $y$ coordinate of the poset point with $x$-coordinate $i$. 
Let $p^x_j = (x_j, j)$ be the unique point with $y$-coordinate $j$, and let
$p^y_i =(i,y_i)$ be the unique point with $x$-coordinate $i$.

Since we generated the
coordinates of the vertices uniformly at random, we can view the $y$ coordinates of the
vertices as $n$ uniformly random samples $s_k$ without replacement from the integer range
$[1,n]$. Thus, the $k$th vertex in the graph is at position $2k, y_{2k}$, given by the
value of the $k$th sample $s_k = y_{2k}$.

Now consider the probability that a point in a particular row $j$ has a certain
$x$ coordinate $x_j = i$:
\[
    Pr(x_j = i) = Pr(y_i = j)
\]
That is, the point in row $j$ has $x$-coordinate $i$ if and only if the unique
point in column $i$ (the $\frac{i}{2}th$ sample) has $y$-coordinate $j$.
Hence, given that $p_j \neq q_i$, the values $x_j$ and $y_i$ are independent.

By construction, for each odd $(i,j)$ in $[3, 2n - 1]$ 
there exists a confluent junction at position $(i,j)$ if and only if
all of the following conditions hold:
\begin{subequations}
    \begin{align} 
        y_{i-1} &< j-1  \\
        y_{i+1} &> j+1  \\
        x_{j-1} &< i-1 \\
        x_{j+1} &> i+1
     \end{align}
 \end{subequations}
Note that since all the coordinates are integers, 
we can equivalently state these conditions as follows
\begin{subequations}
	\begin{align} 
         y_{i-1} < j &\quad \text{and}\quad  y_{i-1} \neq j-1 \label{eqn:ym}\\
         y_{i+1} > j &\quad \text{and}\quad  y_{i+1} \neq j+1 \label{eqn:yp}\\
         x_{j-1} < i &\quad \text{and}\quad  x_{j-1} \neq i-1 \label{eqn:xm}\\
         x_{j+1} > i &\quad \text{and}\quad  x_{j+1} \neq i+1 \label{eqn:xp}
	\end{align}
\end{subequations}

Moreover, the inequality constraint in 
equation~\ref{eqn:ym} is satisfied if and only if the inequality constraint in
equation~\ref{eqn:xm} is satisfied, since there is exactly one vertex in each
even row and column. Likewise, the inequality constraints in
equations~\ref{eqn:yp} and~\ref{eqn:xp} are equivalent. Hence, we need only keep
one of the inequality constraints from each pair of equations, and the inequality
constraints can be equivalently stated: 
$p^x_{j-1} \neq p^y_{i-1}$ and 
$p^x_{j+1} \neq p^y_{i+1}$.

Thus, there exists a junction at $(i,j)$ if and only if
\begin{align*}
     x_{j-1} < i < x_{j+1} &\quad  \text{and} \\
    y_{i-1} < j < y_{i+1}  &\quad  \text{and}  \\
     p^x_{j-1} \neq p^y_{i-1}  &\quad  \text{and} \\ 
    p^x_{j+1} \neq p^y_{i+1}
\end{align*}
That is, the three $x$-coordinates must be in a specific order, 
the three $y$-coordinates must be in a specific order, and there is a $1/n$
fraction of forbidden coordinates in each of the inequality constraints. 

Let $i$ and $j$ be chosen
independently and uniformly at random.
Then, the probability 
that $(i,j)$ has a confluent junction is

\begin{align*}
    & Pr\left(  p^x_{j-1} \neq p^y_{i-1} \right)
     \cdot Pr \left( p^x_{j+1} \neq p^y_{i+1} \mid p^x_{j-1} \neq p^y_{i-1}  \right) \\
    & \cdot Pr\left( x_{j-1} < i < x_{j+1} \mid 
                        p^x_{j-1} \neq p^y_{i-1}, 
                        p^x_{j+1} \neq p^y_{i+1}  
                        \right )\\
    & \cdot Pr\left(y_{i-1} < j < y_{i+1} \mid 
                                p^x_{j-1} \neq p^y_{i-1}, 
                                p^x_{j+1} \neq p^y_{i+1},
                                x_{j-1} < i < x_{j+1}
                                \right) \\
     &= (1 - \frac{1}{n})^2 \cdot\frac{1}{6}\cdot\frac{1}{6}
\end{align*}

Thus, by linearity of
expectation, the total expected
number of confluent junctions over the whole grid, in a confluent drawing of a random two-dimensional partial order, is $\Theta(n^2)$.

\begin{figure}[htp]
    \vskip-1em
\centering
\includegraphics[height=.4\textheight]{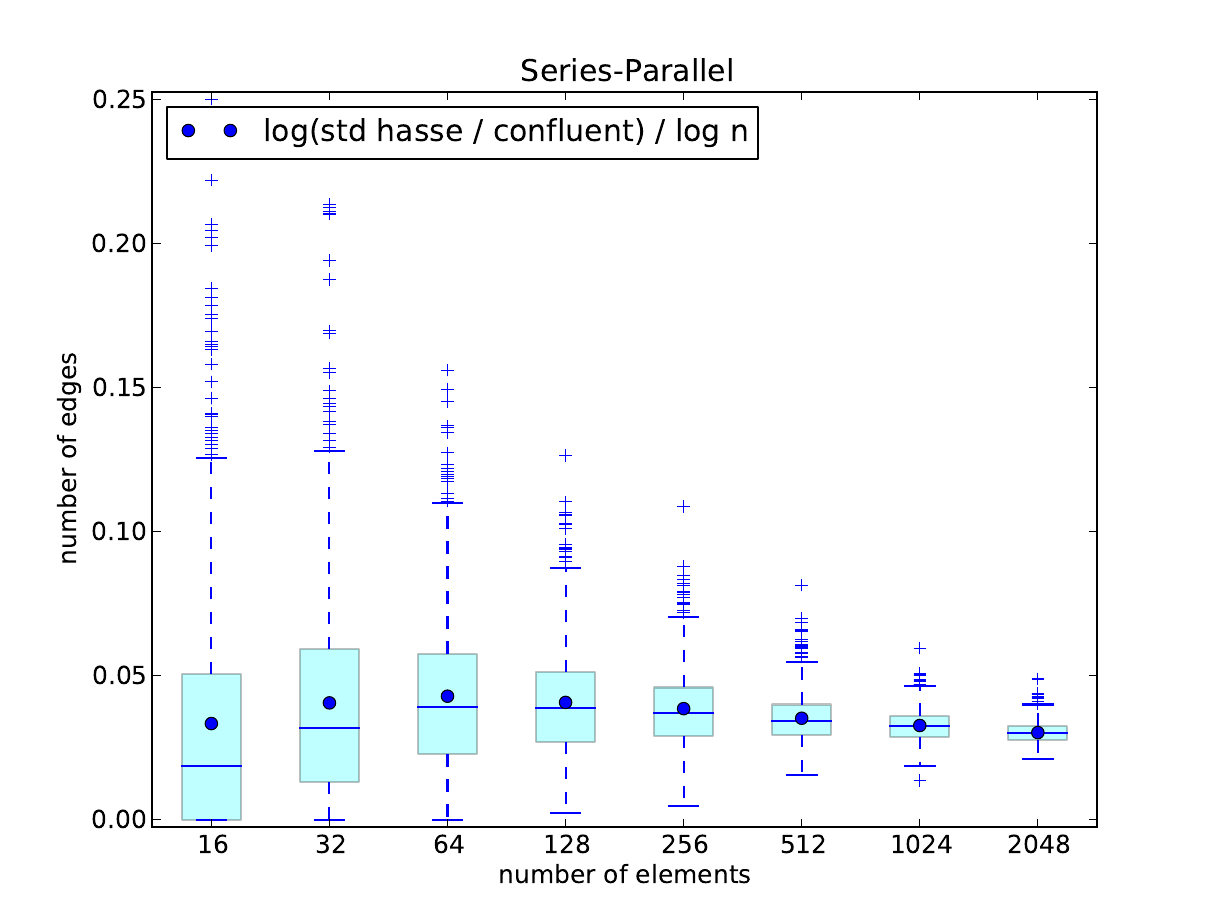}
\caption{\label{fig:ratio-ser-num-ln} 
    A log-log box plot of the ratio of the number of edges in a 
  traditional Hasse diagram to the number of edges in a
  confluent Hasse diagram as a function of the number of vertices
  in upward drawings of series-parallel partial orders. The ratio is normalized
  by dividing by $\log n$.
}
\includegraphics[height=.4\textheight]{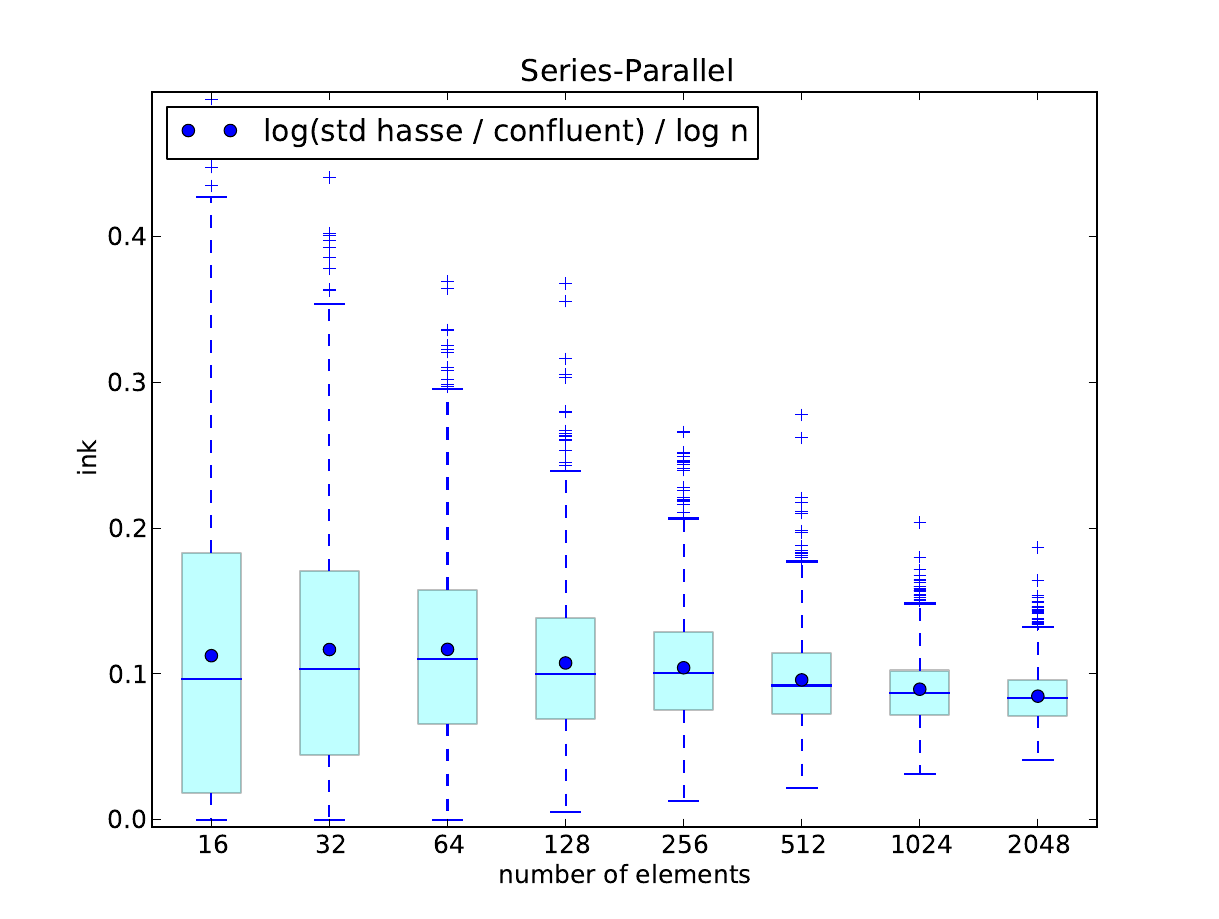}
\caption{\label{fig:ratio-ser-ink-ln} 
    A log-log box plot of the ratio of the ink used in a 
  traditional Hasse diagram to the ink used in a
  confluent Hasse diagram as a function of the number of vertices
  in upward drawings of series-parallel partial orders. The ratio is normalized
  by dividing by $\log n$.
  }
\end{figure}

\begin{figure}[htp]
    \vskip-1em
\centering
\includegraphics[height=.4\textheight]{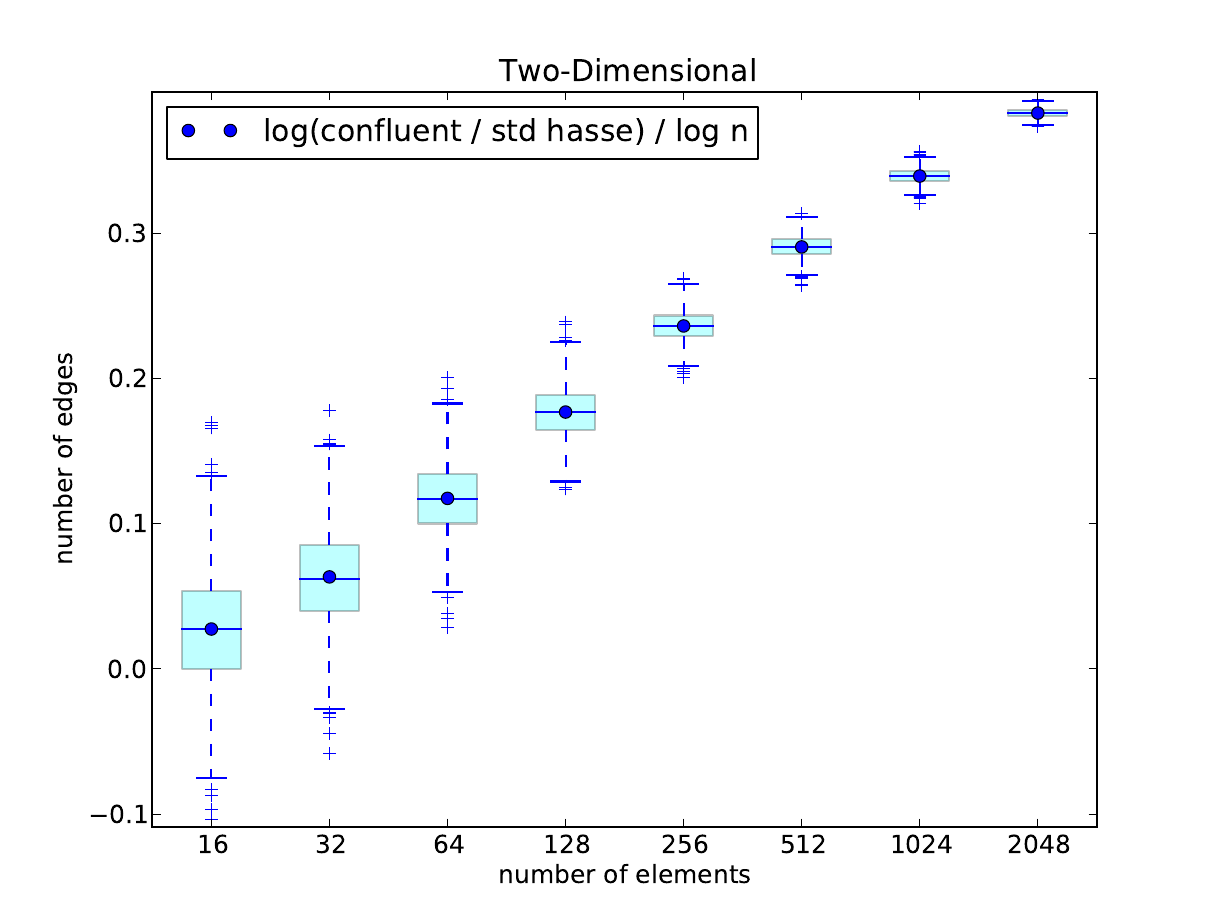}
\caption{\label{fig:ratio-rand-num} 
  A log-log box plot of the ratio of the number edges in a 
  confluent Hasse diagram to the number of edges in a
  traditional Hasse diagram as a function of the number of vertices
  in upward drawings of two-dimensional partial orders. The ratio is normalized
  by dividing by $\log n$.
}
\includegraphics[height=.4\textheight]{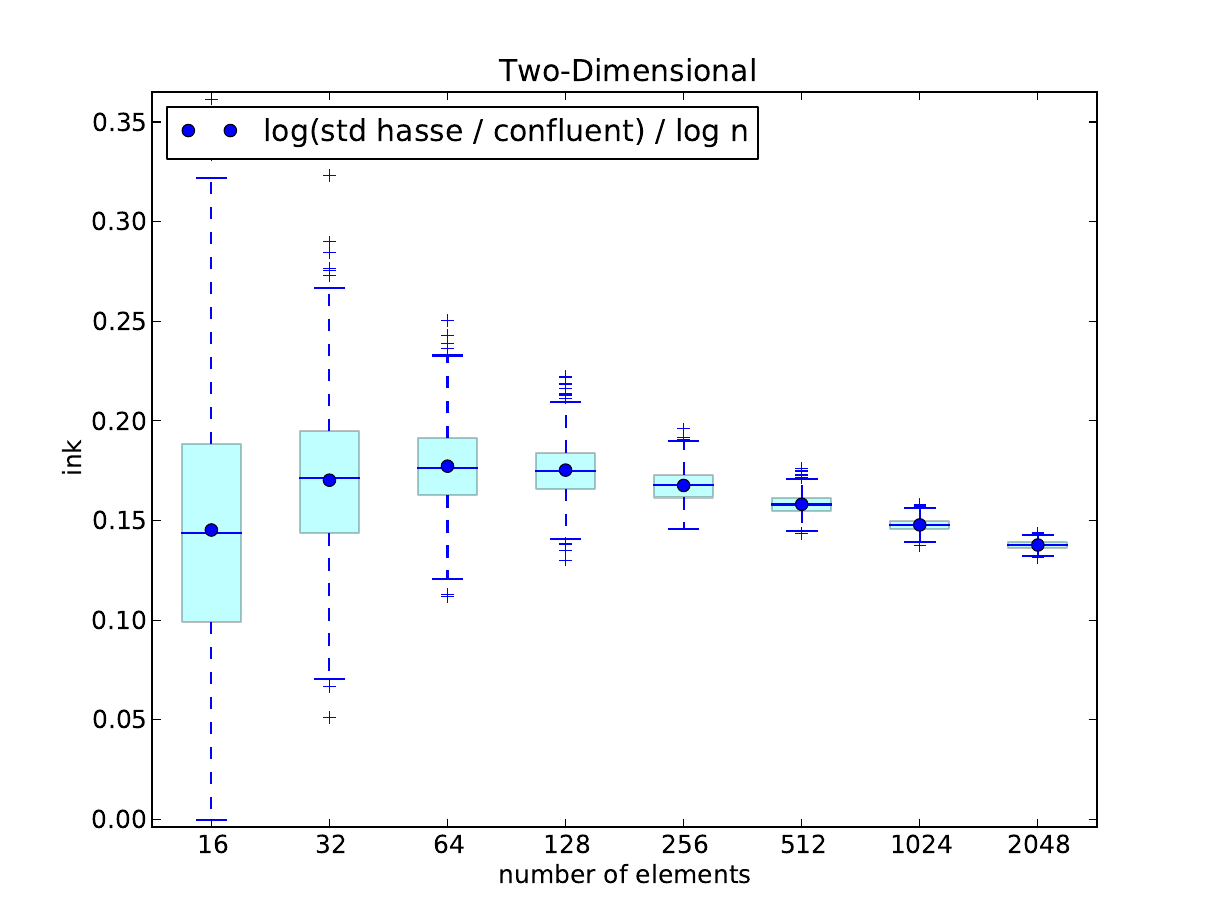}
\caption{\label{fig:ratio-rand-ink-ln} 
    A log-log box plot of the ratio of the ink used in a 
  traditional Hasse diagram to the ink used in a
  confluent Hasse diagram as a function of the number of vertices
  in upward drawings of two-dimensional partial orders. The ratio is normalized
  by dividing by $\log n$.
}
\end{figure}

\section{Conclusions}

We have designed, analyzed, and implemented an algorithm for drawing confluent Hasse diagrams using a minimum number of confluent junctions.
We experimentally verified that confluent diagrams consistently use less ink
than the corresponding traditional Hasse diagrams of both two-dimensional and
series-parallel partial orders. Confluent diagrams of series-parallel partial
orders also use fewer edges. Confluent diagrams of two-dimensional partial
orders often use substantially more edges than the corresponding traditional
Hasse diagram. However, the larger number of edges used is required by the
larger number of confluent junctions required to address all the edge crossings
in these graphs. The result is a drawing with more edges but substantially less visual clutter.

Upward planarity may be tested even for non-st-planar graphs that have only one source or one sink; can similar conditions be extended to the case of upward confluent drawings? 
Can we efficiently find upward planar drawings of graphs that are not transitively reduced? If a partially ordered set must be drawn with crossings, can we use confluence in a principled way to keep the number of crossings small? We leave these questions to future research.

\clearpage 
\subsubsection*{Acknowledgements}

This work was supported in part by NSF grants
0830403 and 1217322 and by the Office of Naval Research under grant
N00014-08-1-1015.

\ifFull
{
\raggedright
\bibliographystyle{abuser} 
\bibliography{poset}
}\fi

\end{document}